\documentclass[a4paper,11pt]{article}
\usepackage[margin=1.0
in]{geometry}
\usepackage[utf8]{inputenc}
\usepackage{amsmath}
\usepackage{amsfonts}
\usepackage{amssymb}
\usepackage{amsthm}
\usepackage{graphicx}

\newtheorem{definition}{Definition}
\newtheorem{theorem}{Theorem}
\newtheorem{corollary}{Corollary}
\newtheorem{lemma}{Lemma}
\newtheorem{remark}{Remark}

\title{Edge complexity of geometric graphs on convex independent point sets}
\author{Abhijeet Khopkar}

\begin{document}

\maketitle

\begin{abstract}
In this paper, we focus on a generalized version of Gabriel graphs known as Locally Gabriel graphs
($LGGs$) and Unit distance graphs ($UDGs$) on convexly independent point sets. $UDGs$ are sub graphs of $LGGs$.
We give a simpler proof for the claim that $LGGs$ on convex independent point sets have $2n \log n + O(n)$ edges.
Then we prove that unit distance graphs on convex independent point sets have $O(n)$ edges improving the previous known
bound of $O(n \log n)$.
\end{abstract}

\section{Introduction}
Tur{\'a}n type problems have a rich history in graph theory. Tur{\'a}n's classical problem is to find the maximum number of edges $ex(n,H)$  a graph (on $n$ vertices) can have without containing
a subgraph isomorphic to $H$ (refer to Tur{\'a}n's theorem ~\cite{turan}). A simple example is that a graph not containing any cycle (acyclic graph) has linear number of edges.
These type of problems have been extensively studied for the geometric graphs~\cite{suk,toth,tth,crss,pv}. 
Various geometric graphs have been studied for special point sets like points on a uniform grid and the convex point sets. 
Tur{\'a}n type problems have also been studied on the geometric graphs when all the points are in convex position~\cite{Capoy,ram,perles,sas,fox}.
The vertices in convex position also provide a cyclic ordering on the vertices. Thus, an obvious technique to explore these problems is by extracting abstract combinatorial structures from geometric conditions
and the order on the vertices.
Similar problems have been addressed by an alternate approach of counting the maximum number of 1s in a 0-1 matrix 
where some sub matrices are forbidden. A 0-1 matrix can represent the adjacency matrix of a bipartite graph with an ordering on the vertices in both the partitions. The maximum number of 1s corresponds to the maximum number of
edges the graph can have. The problem of counting 1s in a 0-1 matrix for various forbidden sub matrices is explored extensively~\cite{furedi,H92,pet10,Keszegh09}.

For many geometric graphs, the maximum number of edges is determined by characterizing a forbidden subgraph by some geometric restriction, for examples refer to~\cite{furedi,yyy,ps04}.
In this paper, we study {\it Unit distance graphs} and {\it Locally Gabriel graphs} on convex point sets. We characterize some forbidden patterns in these graphs and use them to study the maximum number of edges.
\subsection{Unit distance graphs}
Unit distance graphs ($UDGs$)\footnote{Not to be confused with the unit disk graphs} are well studied geometric graphs. In these graphs an edge exists between two points if and only if the Euclidean distance between the points is unity. 
\begin{definition}
 A geometric graph $G=(V,E)$ is called a unit distance graph provided that for any two vertices $v_1,v_2 \in V$, the edge $(v_1,v_2) \in E$ if and only if the Euclidean
distance between $v_1$ and $v_2$ is exactly unity.
\end{definition}
UDGs have been studied extensively for various properties including their the number of edges. The upper bound and the lower bound for the
the maximum number of edges in Unit distance graphs (on $n$ points in $\mathcal{R}^2$) are $O(n^\frac{4}{3})$~\cite{bb} and $\Omega(n^\frac{1}{c\log\log n})$ (for a suitable constant $c$) respectively~\cite{Erdudg}
Erd\H{o}s showed an upper bound of $O(n^\frac{3}{2})$~\cite{Erdudg}. The bound was first improved to $o(n^\frac{3}{2})$~\cite{szem}, then improved to $n^{1.44\ldots}$~\cite{beck}.
Finally, the best known upper bound of $O(n^\frac{4}{3})$ was obtained in~\cite{bb}. Alternate proofs for the same bound were given in~\cite{Szk,pach}.
Bridging the gap in these bounds has been a long time open problem. Unit distance graphs have also been studied for various special point sets most notably the case when 
all the points are in convex position. The best known upper bound for the number of edges in a unit distance graph on a convex point set with $n$ points is $O(n \log n)$.
The first proof for this upper bound was given by Zolt{\'a}n F{\"u}redi~\cite{furedi}. The proof is motivated by characterizing a 3 $\times$ 2 sub matrix that is forbidden in a 0-1 matrix.
The sub matrix is motivated by the definition of $UDGs$ and the convexity of the point set.
It was shown that any such $a \times b$ matrix has at most $a+(a+b) \lfloor \log_{2}b \rfloor$ number of 1s. The argument can be easily extended to show that the adjacency matrix of a $UDG$ on a convex point set of size $n$ has $O(n\log n)$ number of 1s that corresponds to the total number of edges.
Peter Bra{\ss} and J{\'a}nos Pach provided a simpler alternative proof using a simple divide and conquer technique~\cite{bp}.
Another proof for the same bound using another forbidden pattern supplemented by a divide and conquer technique was given in~\cite{sas}. The best known lower bound on the number of unit distances in a convex point set is $2n-7$ for $n$ vertices~\cite{han}.
Bridging the gap in the bounds for this special case has also been an interesting open problem. Some interesting questions on the properties of unit distances in a convex point set are studied in~\cite{Erd86,Fb92}.
Unit distance graphs have also been studied for more special types of convex point sets, e.g centrally symmetric convex point set. Unit distance graphs on
centrally symmetric convex point sets have $O(n)$ edges~\cite{csc}.

\subsection{Locally Gabriel Graphs}
Gabriel and Sokal \cite{ggg} defined a Gabriel graph as follows:
\begin{definition} \label{gdef}
A geometric graph $G=(V,E)$ is called a Gabriel graph if the following condition holds:
For any $u,v \in V$, an edge $(u,v) \in E$ if and only if the disk with $\overline{uv}$ as diameter does not contain any other point of $V$.
\end{definition}

Motivated by applications in wireless routing, Kapoor and Li~\cite{yyy} proposed a relaxed version of Gabriel graphs known as $k$-locally Gabriel graphs.
These structures have been studied for the bounds on the number of edges in \cite{yyy,ps04}.
In this paper, we focus on 1-locally Gabriel graphs and call them {\em Locally Gabriel Graphs} ($LGG$s).
\begin{definition}\label{lggdef}
A geometric graph $G=(V,E)$ is called a Locally Gabriel Graph if for every $(u,v) \in E$, the disk with the line segment $\overline{uv}$ as 
diameter does not contain any neighbor of $u$ or $v$ in $G$.
\end{definition}
The above definition implies that two edges $(u,v)$ and $(u,w)$ where $u,v,w \in V$ {\it conflict} with each other if $\angle uwv~\ge~\frac{\pi}{2}$ or $\angle uvw \ge \frac{\pi}{2}$
and cannot co-exist in an $LGG$, i.e. it is not possible to satisfy the condition  $(u,v) \in E$ and $(u,w) \in E$.
Conversely if edges $(u,v)$ and $(u,w)$ co-exist in an $LGG$,
then $\angle uwv~<~\frac{\pi}{2}$ and $\angle uvw < \frac{\pi}{2}$. We call this condition as {\it LGG constraint}.
In this paper, we explore these graphs on convex independent point sets.
Let us highlight an important property to be noted for this paper.
\begin{remark}\label{uisl}
A Unit distance graph ($UDG$) is also a Locally Gabriel Graph ($LGG$).
\end{remark}
Note that if two line segments $\overline{uv}$ and $\overline{uw}$ of the same length share a common end point $u$, then $\angle uvw < \frac{\pi}{2}$
and $\angle uwv < \frac{\pi}{2}$. Thus, two edges from a vertex in $UDG$ satisfy {\it LGG constraint} leading to Remark~\ref{uisl}.
\subsection{Preliminaries and Notations}\label{pn}
A graph is called an ordered graph when the vertex set of the graph has a total order on it. We consider a bipartite graph when the vertex set in each partition has a total order on its vertices.
Formally, an ordered bipartite graph is $G = (U,V,<_U,<_V,E)$. There are two linear ordered sets $(U,<_U)$ and $(V,<_V)$ of the vertices and $E \subseteq U \times V$.
We define a special family of such bipartite graphs where some structures in these graphs are forbidden. A {\it path} in a graph represents a sequence of the edges s.t. two consecutive edges share a vertex.
A path can be represented as a set of edges.

\begin{definition}
 A path $P$ in the ordered bipartite graph  $G=(U,V,E)$ that  visits the vertices in $U$ and $V$ in the order $u_1,u_2,\ldots,u_k$ and $v_1,v_2,\ldots,v_l$ respectively,
is called a forward path if either $u_1 < u_2 \ldots <u_k$ and $v_1 < v_2 \ldots < v_l$ or $u_1 > u_2 \ldots > u_k$ and $v_1 > v_2 \ldots > v_l$.
\end{definition}
An ordered set represented as $\langle u_1,u_2 \rangle$ for $u_1, u_2 \in U$ denotes all the vertices $u_i$ s.t. $u_1 \le u_i \le u_2$.
Note that this ordered set includes $u_1$ and $u_2$ as well.
Similarly, an ordered set $\langle v_1,v_2\rangle$ for $v_1,v_2 \in V$ denotes all the vertices $v_i$ s.t. $v_1 \le v_i \le v_2$.
The {\it range} of a forward path $P$ that passes through the vertices $u_{a}, u_{b}, v_{c}$ and $v_{d}$ is denoted as $\{\langle u_{a},u_{b}\rangle,\langle v_{c},v_{d}\rangle\}$,
represents all the vertices (assume that $u_{a} < u_{b}$ and $v_{c} < v_{d}$) $u_i$ and $v_j$ s.t. $u_{a} \le u_i \le u_{b}$ and $v_{c} \le v_j \le v_{d}$.
An edge $(u_{a}, v_j) ($resp. $(v_{c},u_i))$ is called a {\it back edge} to the forward path $P$ if $v_j \in {\langle v_{c},v_{d}\rangle} ($resp. $u_i \in \langle u_{a},u_{b}\rangle)$
and $u_i > u_{a+1} ($resp. $v_j > v_{c+1})$ where $u_{a+1} \in U ($resp. $v_{c+1} \in V)$ is a non terminal vertex in $P$, i.e. this vertex has edges incident to two vertices in $P$.
Refer to Figure~\ref{fig1} for a pictorial representation.
\begin{figure}[ht]
\begin{center}
 \includegraphics[scale=0.4]{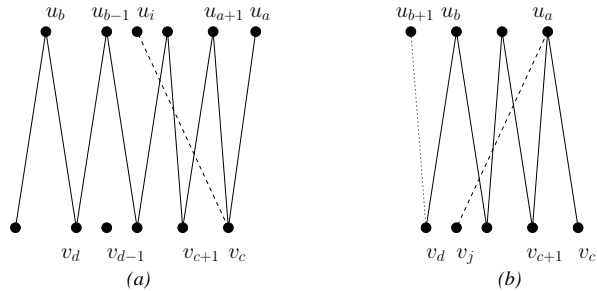}
 \caption{A path and back edges in a $PRBG$}
\label{fig1}
\end{center}
\end{figure}
\begin{definition}\label{prp}
 An ordered bipartite graph $G=(U,V,E)$ is said to satisfy the path restricted property if for any forward path $P$ in $G$, there exists no back edge $e \in E$ to $P$.
\end{definition}
A path-restricted ordered bipartite graph ($PRBG$) is an ordered bipartite graph that satisfies the path restricted property.
Note that a $PRBG$ follows the constraint presented by F{\"u}redi~\cite{furedi}, where it was proved that any bipartite graph
following this constraint has $O(n \log n)$ edges. It also implies that a $PRBG$ on $n$ vertices has $O(n \log n)$ edges.

To represent these graphs pictorially, for convenience the vertices are placed from right to left in the increasing order.
Note that subscripts to the vertices are also mostly increasing from right to left corresponding to the increasing order.
However, at some places for better readability of the arguments subscripts are not necessarily increasing from right to left that
should not be confused with the order of vertices. In a pictorial representation, order is always increasing from right to left.
\subsection{Our Contributions}
We establish a relationship between UDGs/LGGS on convex point sets and the path restricted ordered bipartite graphs.
The following are the main results presented in this paper.
\begin{itemize}
 \item We prove that $UDGs$ on convex point sets have $O(n)$ edges.
 \item We give an alternate and simpler proof (compared to the F{\"u}redi's proof~\cite{furedi}) that a path restricted ordered bipartite graph on $n$ points has at most $n \log n + O(n)$ edges.
It also proves that $LGGs$ on convex point sets have at most $2n \log n + O(n)$ edges.
\end{itemize}
\section{Obtaining $PRBGs$ from $UDGs/LGGs$}\label{sc2}
In this section, we show that a $UDG/LGG$ on a convex point set can be decomposed into two $PRBGs$ by removing at most linear number of edges.
First, we focus on some fundamental properties of the unit distance graphs on a convex point set. Two points $p_i$ and $p_j$ in a convex point set $P$ are called antipodal points
if there exist two parallel lines $\ell_i$ passing through $p_i$ and $\ell_j$ through $p_j$, such that all other points in $P$ are contained between $\ell_i$ and $\ell_j$.
\begin{lemma}\label{ap}
\cite{bp}Let $G = (P,E)$ be a unit distance graph on convex point set $P$. If $p_i \in P$ and $p_j \in P$ are two antipodal points, 
then all but at most $2|P|$ edges of $G$ cross the line $\overline{p_ip_j}$. 
\end{lemma}

Let $p_1$ and $p_2$ be two antipodal points in the given convex point set $P$ as shown in Figure~\ref{ufig1}. Let
$U$ and $V$ be the sets of the points on the opposite sides of the line $\overline{p_1p_2}$.
Let the vertices in $U$ and $V$ be $u_1,u_2,\ldots u_{n_1}$ and $v_1,v_2,\ldots v_{n_2}$ respectively (from right to left).
Remove all the edges that do not cross the line $\overline{p_ip_j}$. Let $E'$ be the set of the remaining edges.
Consider the bipartite graph $G = (U,V,E')$. $E'$ is divided into two disjoint sets $E_1$ and $E_2$ by the following rule.
Consider an edge $(u_i,v_j)$, let $v_{j-1}$ and $v_{j+1}$ be the adjacent vertices to $v_j$ in $V$ on left and right side respectively as shown in Figure~\ref{ufig2}.
By convexity, it can be observed that at least one of $\angle u_iv_jv_{j+1}$ or $\angle uv_jv_{j-1}$ is acute.
If $\angle u_iv_jv_{j+1}$ is acute then put the edge $(u_i,v_j)$ in $E_1$ else if $\angle u_iv_jv_{j-1}$ is acute then put the edge $(u_i,v_j)$ in $E_2$.
If both the angles are acute, then the edge can be put arbitrarily in either $E_1$ or $E_2$.
In the graph $G_1 = (U,V,E_1)$, the vertices are ordered as $u_1 < u_2 < \ldots u_{n_1}$ in $U$ and
$v_1 < v_2 < \ldots v_{n_2}$ in $V$. The ordering is reversed in the graph $G_2 = (U,V, E_2)$.
\begin{remark}
 In $G_1$ and $G_2$, no two edges intersect in a forward path.
\end{remark}
\begin{figure}[ht]
\begin{minipage}[b]{0.5\linewidth}
\centering
\includegraphics[scale=0.5]{udg10}
\caption{Antipodal points in a convex point set}
\label{ufig1}
\end{minipage}
\hspace{0.5cm}
\begin{minipage}[b]{0.4\linewidth}
\centering
\includegraphics[scale=0.55]{udg0}
\caption{Partition of the edges}
\label{ufig2}
\end{minipage}
\end{figure}
Let us remove the extreme left edge incident to every vertex $v \in V$ from $G_1$ and call this process {\it left trimming}.
Formally, if $v$ has edges incident to $u_1,u_2,\ldots,u_k$ such that $u_k > \ldots > u_1$, then the edge $(v,u_k)$ is removed from$G_1$,
the resultant graph is called $G'_1$. Similarly, by {\it right trimming} 
the extreme right edge for every vertex $v \in V$ in $G_2$ is removed to obtain the graph $G'_2$. Let $G_{UDG}$ denote the class of the ordered bipartite graphs,
consisting of the graphs $G'_1$ and $G'_2$ that are obtained from the unit distance graphs. It can be assumed w.l.o.g. that $|V| \le |U|$.

Consider the Locally Gabriel graphs on a convex point set. It was observed in ~\cite{gclgg} that Lemma~\ref{ap} holds true for Locally Gabriel graphs too.
Therefore, a bipartition can be obtained similarly by dividing a convex point set along two antipodal points. Consider the bipartite graph between the two partitions.
Similar to $G_{UDG}$, a new graph class $G_{LGG}$ can be defined. The procedure to obtain a graph in $G_{UDG}$ (from the $UDG$ on a convex point set) can also be applied to an $LGG$ on a convex point set to obtain a graph in $G_{LGG}$.

We show that the graphs in $G_{UDG}$ and $G_{LGG}$ are path-restricted ordered bipartite graphs in Lemma~\ref{lggpbg}. Thus, a $UDG/LGG$ on convex a point set can be decomposed into two $PRBGs$ by removing at most $3n$ edges where $n = n_1 + n_2$
($2n$ edges from Lemma~\ref{ap} and $n$ edges by deleting extreme edges as mentioned above).

\begin{lemma}\label{lggpbg}
 Any graph $G = (U,V,E)$ in $G_{LGG}$ satisfies the path restricted property. Therefore, $G$ is a $PRBG$.
\end{lemma}
\begin{proof}
Recall that two graphs in class $G_{LGG}$ ($G'_1$ and $G'_2$) are obtained from an $LGG$ on a convex point set. We prove
the Lemma for $G'_1$. A symmetrically opposite argument can be given to prove the stated Lemma for $G'_2$ type of graphs.
We show that if $P$ is a forward path in $G = (U,V,E)$ with the range $R_P=~\{\langle u_a,u_b\rangle,\langle v_c,v_d\rangle\}$,
then there does not exist a back edge $(u_i,v_c) \in E$ where $u_i \in \langle u_a,u_b\rangle$.
The path $P$ and the concerned vertices along with the edges are shown in Figure~\ref{f3}($a$).
Let $v_{d-1} \in V$ be the vertex preceding $v_d$ in $V$. Note that $(u_b,v_d)$ is an edge in $P$.
Now $\angle u_bv_dv_{d-1} < \frac{\pi}{2}$ (by the definition of $G_{LGG}$). By convexity, it can be further inferred that
$\angle u_bv_dv_c < \frac{\pi}{2}$. Let $u_{b-1} \in U$ be the vertex in $P$ with an edge incident to $v_d$ (apart from 
$u_{b}$) and $v_{c+1} \in V$ be the vertex that immediately succeeds to $v_c$ in $P$. By the definition~\ref{lggdef} of
$LGGs$, $\angle v_du_bu_{b-1}, \angle u_av_cv_{c+1} < \frac{\pi}{2}$. By convexity, $\angle v_du_bu_a, \angle u_av_cv_d <
\frac{\pi}{2}$.Thus, in the quadrilateral $u_av_cv_du_b$, $\angle u_bu_av_c$ must be greater than $\frac{\pi}{2}$. By convexity, $\angle u_iu_av_c > \frac{\pi}{2}$.
Thus, the edge $(u_i,v_c)$ and $(u_a,v_c)$ conflict with each other. Therefore, the edges $(u_i,v_c)$ cannot exist in $G$
for any $u_i \in \langle u_a,u_b\rangle$.
\begin{figure}[ht]
\begin{center}
 \includegraphics[scale=0.4]{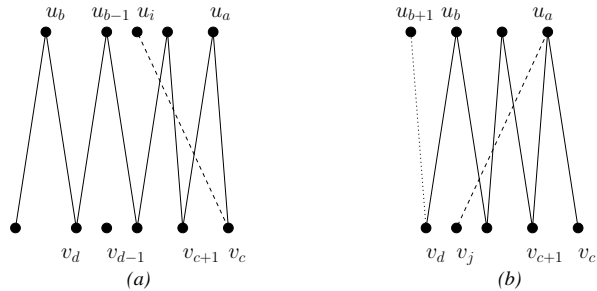}
 \caption{$G_{LGG}$ has path restricted properties}
\label{f3}
\end{center}
\end{figure}

Recall that the leftmost edge incident to every vertex $v \in V$ is deleted in the graph $G_1=(U,V,E_1)$ to obtain a $G_{LGG}$ ({\it left trimming} in Section~\ref{sc2}).
Let $(u_{b+1}$ be the leftmost vertex with an edge incident to $v_d$. Note that the edge $(u_{b+1},v_d)$ is present in $G_1$ but not
present in $G'_1$. Now similar arguments lead to prove that a back edge is forbidden. If there exists a back edge $(u_a,v_j) \in E$ where $v_j \in \langle v_c,v_d\rangle$,
then in the convex quadrilateral $u_av_{c+1}u_dv_{b+1}$, all the internal angles are acute.
Therefore, if $P$ is a forward path in $G_{LGG} = (U,V,E)$ with the range $R_P = \{\langle u_a,u_b\rangle,\langle v_c,v_d\rangle\}$, then there does not exist a
back edge $(u_a,v_j) \in E$ where $v_j \in \langle v_c,v_d\rangle$ (refer to Figure~\ref{f3}($b$)).

Thus, any graph in $G_{LGG}$ satisfies the path restricted property. Therefore, a graph in the class $G_{LGG}$ is also a $PRBG$.
\end{proof}

By Remark~\ref{uisl} a Unit distance graph is also a Locally Gabriel graph. Therefore, any graph in the class $G_{UDG}$ also
belongs to the class $G_{LGG}$.
\begin{lemma}
 Any graph $G = (U,V,E)$ in $G_{UDG}$ satisfies the path restricted property. Therefore, $G$ is a $PRBG$.
\end{lemma}
\section{Properties of the path restricted ordered bipartite graphs}\label{sc3}
First let us describe some properties of $PRBGs$ arising from the {\it path restricted} property. A forward path with
range $\{\langle u_a,u_b\rangle,\langle v_c,v_d\rangle\}$ can be said to emerge or start from from any of the vertices $u_a, u_b, v_c$ or $v_d$.
\begin{corollary} \label{rem1}
 In a path restricted ordered bipartite graph, two forward paths emerging from a vertex in the same direction never meet each other.
\end{corollary}
\begin{corollary} \label{rem2}
From any vertex (not in the forward path $P$), only one edge can be incident to the vertices in the forward path $P$.
\end{corollary}
Let us consider all the forward paths starting from a vertex. These paths could be classified into two sets. The first set consists of all the forward paths visiting to the lower ordered vertices (rightwards) and 
the second set consists of all the forward paths visiting to the higher ordered vertices (leftwards).
Let us consider first the set of the paths visiting rightwards. 
From the subsequent vertices on these paths, multiple paths can emerge visiting the vertices rightwards. These paths never meet with each other (refer to Corollary~\ref{rem1}).
Thus, these forward paths emerging from a vertex form a tree. Let $T_r(u)$ denotes such a tree originating from $u$.
Similarly, $T_l(u)$ denotes a tree that consists of all the forward paths originating from $u$ visiting the higher ordered vertices (leftwards).
\begin{lemma}
For any vertex $v$ in a $PRBG$ $G=(U,V,E)$, the subgraph induced by the vertices of $T_r(v)$ has $n-1$ edges where $n$ is the number of vertices spanned by $T_r(v)$.
\end{lemma}
\begin{proof}
We show that for any vertex $v$ (let $v \in V$ w.l.o.g) in a $PRBG$, the subgraph induced by the vertices in $T_r(v)$ does not have any edge but the edges in $T_r(v)$.
On the contrary, let there exists an edge $(u_i,v_i) \in E$  s.t. this edge is not present in $T_r(v)$ and the vertices ($u_i \in U$ and $v_i \in V$) are spanned by $T_r(v)$. Recall that two forward paths emerging from a vertex
in the same direction never meet again (refer to Corollary~\ref{rem1}). Therefore, the edge $(u_i,v_i)$ does not belong to any forward path emerging from $v$. Let $u_j \in U$ be the vertex with the highest order incident to $v$.
Note that $u_i$ and $u_j$ are not the same vertices and $u_i < u_j$ (refer to Figure~\ref{u3}(a)). $u_i$ cannot have an edge incident to $v$, otherwise the edge $(u_i,v_i)$ belongs to a forward path originating from $v$ as shown in Figure~\ref{u3}(b).
But there exists a forward path passing through $v$ and $u_i$. Let $v_{i'} \in V$ be the vertex preceding $u_i$ in the forward path from $v$ to $u_i$. Observe
that $v_{i'} < v_i$. Thus, there exists a forward path with the range $\{\langle u_i,u_j\rangle,\langle v_{i'},v\rangle\}$. Therefore, the back edge $(u_i,v_i)$ is forbidden by the definition of $PRBGs$.
Thus, it leads to a contradiction to the assumption that there exists an edge between $u_i$ and $v_i$.
\end{proof}
\begin{figure}[ht]
\begin{minipage}[b]{0.5\linewidth}
\begin{center}
\includegraphics[scale=0.35]{udg14}
\caption{Edge $(u_i,v_i)$ is forbidden}
\label{u3}
\end{center}
\end{minipage}
\hspace{0.5cm}
\begin{minipage}[b]{0.4\linewidth}
\centering
\includegraphics[scale=0.6]{udg3}
\caption{Edges in $T_l(v)$}
\label{u4}
\end{minipage}
\end{figure}
\begin{lemma}
 For any vertex $v$ in a $PRBG$ $G = (U,V,E)$, all the forward paths in $T_l(v)$ have disjoint ranges.
\end{lemma}
\proof
Let us assume w.l.o.g. that $v \in V$. Consider two forward paths in $T_l(v)$ originating from $v$. Consider a path $P_1 = (v, u_1, v_1, \ldots)$ as shown in Figure~\ref{u4}.
Also consider the path $P_2 = (v,u_2,v_2, \ldots)$ where $v_1 < v_2$ (for $v_1,v_2 \in V$). Observe that there is a restriction that $u_1 > u_2$ ($u_1,u_2 \in U$), otherwise
the edge $(u_1,v_1)$ is forbidden by the {\it path restricted property}. Similarly, let $u_i \in U$ and $v_i \in V$ be the successive vertices in $P_1$
and let $u_j \in U$ and $v_j \in V$ be the successive vertices in $P_2$. By the {\it path restricted property}, it can be observed that if $v_i < v_j$, then $u_j < u_i$.
Therefore, the ranges of the paths $P_1$ and $P_2$ are disjoint.
\qed
\section{Number of edges in path restricted ordered bipartite graphs}\label{sc4}
In this section, we study $PRBGs$ for the maximum number of edges they can have. Note that motivated by geometric
orientation we use terms left and right to describe relative order of vertices interchange-
ably with higher and lower order of the vertices.
\begin{lemma}[Crossing lemma]\label{pl}
 Consider a $PRBG$ $G = (U, V, E)$ with a separator line $\ell$ partitioning $U$ (resp. $V$) into disjoint subsets $U_1$ and $U_2$ (resp. $V_1$ and $V_2$) 
 such that all the vertices in $U_1$ and $V_1$ are placed to the left of $\ell$ and all the vertices in $U_2$ and $V_2$ are placed to the right of $\ell$.
\begin{enumerate}
 \item If every vertex in $U_1$ has an edge incident to it with the other endpoint in $V_1$, then the number of edges between $U_1$ and $V_2$ (crossing $\ell$) is at most $|U_1| + |V_2|$.
 \item If every vertex in $V_1$ has an edge incident to it with the other endpoint in $U_1$, then the number of edges between $V_1$ and $U_2$ (crossing $\ell$) is at most $|V_1| + |U_2|$.
\end{enumerate}
\end{lemma}
\begin{figure}[h]
\begin{center}
 \includegraphics [scale=0.4]{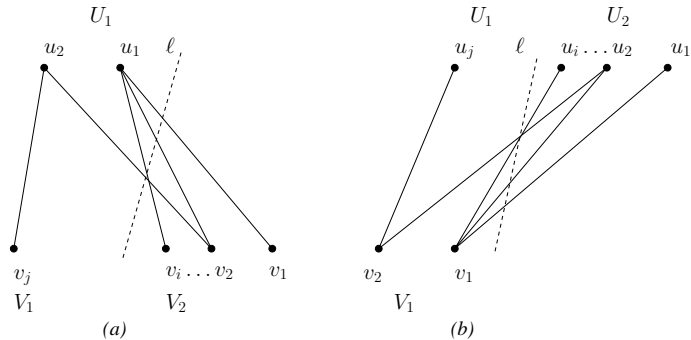}
 \caption{Edges across a partition line}
 \label{u9}
\end{center}
\end{figure}
\proof
An edge crossing the partition line $\ell$ is called the crossing edge.
Let us consider only the vertices (in either of $U_1, U_2, V_1$ and $V_2$) that have more than one crossing edges incident to them.
We give unit charge to all the vertices initially. A vertex can consume its charge to count for an edge.
We show that if every vertex is charged for the leftmost crossing edge incident to it, then all the edges are counted. 

Consider the rightmost vertex $u_1 \in U_1$ (the vertex with the least order in $U_1$) that has crossing edges incident to the vertices $v_1,v_2, \ldots, v_k$ as shown in Figure~\ref{u9}($a$).
We show that any of these vertices except $v_i$ cannot have an edge incident to a vertex in $U_1$ placed to the left of $u_1$ (higher order than $u_1$).
Let us assume on the contrary that $v_2$ has such an edge incident to the vertex $u_2$. 
By assumption $u_2$ has an edge incident to a vertex in $V_1$ (say $v_j \in V_1$), the edge does not intersect $\ell$ and it is placed to the left of it. Since, $v_i$ is placed to the right of $\ell$,
there exists a forward path with the range $\{\langle u_1,u_2\rangle,\langle v_2,v_j\rangle\}$ and the back edge $(u_1,v_i)$ is forbidden by the {\it path restricted property} since $v_i \in \langle v_2,v_j\rangle$. Thus, it contradicts to
the assumption that $v_2$ has an edge incident to $u_2$.
Since $u_1$ is the rightmost vertex in $U_1$ (with the least order in $U_1$), the vertices $v_1, \ldots, v_i$ except $v_i$ have only one crossing edge incident to them. These vertices consume their own charge to count the corresponding edges.
$u_1$ consumes its charge for the edge $(u_1,v_i)$.
Note that all the crossing edges incident to $u_1$ and its adjacent vertices across $\ell$ (except $v_1$) are counted. Also note that the charge of $v_i$ is still not consumed.
Now this charging scheme can be applied to the next vertex to the left of $u_1$. Subsequently, this procedure can be applied to all the vertices in $U_1$ from right to left and all the edges are counted.
Thus, if each vertex in $U_1$ and $V_2$ consumes its charge to count the leftmost edge incident to it, all the edges between $U_1$ and $V_2$ are counted.

Similarly for the proof of (2), if a vertex $v_1 \in V_1$ that has crossing edges incident to the vertices $u_1,u_2, \ldots u_i$ as shown in Figure~\ref{u9}($b$), then the vertices $u_1, \ldots, u_i$ except $u_i$
cannot have an incident to a vertex in $V_1$ placed to the left of $v_1$. 
A similar argument can be made to  show that if each vertex in $V_1$ and $U_2$ consumes its charge to count the leftmost edge incident to it, then all the edges between $V_1$ and $U_2$ are counted.
\qed
\begin{corollary}\label{plc}
Consider the partition of a $PRBG$ as described in Lemma~\ref{pl}. No two vertices in $U_2$ (resp. $V_2$) with edges incident to $V_1$ (resp. $U_1$), i.e. to the vertices left of partition line
can have an edge incident to a common vertex in $V_2$ (resp. $U_2$), i.e. to a vertex at the right of partition line.
\end{corollary}
\begin{corollary}\label{plt}
Consider the partition of a $PRBG$ as described in Lemma~\ref{pl}. Let $G_1 = (U_1, V_1,E_1)$ be a graph with all the vertices to the left of $\ell$.
Let $V'_2 \subseteq V_2$ (resp. $U'_2 \subseteq U_2$) be the set of vertices such that $\forall v \in V'_2$ (resp. $\forall u \in U'_2$) has at least two edges incident to $U_1$ (resp. $V_1$), then the number of edges
between $U_1$ (resp. $V_1$) and $V'_2$ (resp. $U'_2$) is $O|U_1|$ (resp. $O|V_1|$).
\end{corollary}
\proof
The edges between $U_1$ and $V'_2$ (resp. $V_1$ and $U'_2$) are crossing $\ell$. Note that the subgraph induced by these edges is a forest of right
trees, i.e. each component is $T_r(u)$ for some $u \in U_1$ (resp. $T_r(v)$ for some $v \in V_1$). Since each vertex in $U'_2$ {resp. $V'_2$) has
degree of at least two, it implies that all the leaves in these trees are in $U_1$ (resp. $V_1$). Thus, number of such edges is $O|U_1|$ (resp. $O|V_1|$).
\qed
\begin{theorem} \label{T2}
Any path restricted ordered bipartite graph $G = (U,V,E)$ has at most $n \log n + O(n)$ edges where $n = |U| + |V|$. The bound is tight as there exists a path restricted ordered bipartite graph on $n$ vertices with $\Omega(n \log n)$ edges.
\end{theorem}
\proof
We propose a simple divide and conquer technique to get the desired bound. A partition line $\ell$ is drawn dividing the vertices into two halves. Now, we divide the vertices into two disjoint subsets $S_1$
and $S_2$ as shown in Figure~\ref{u5}. All the vertices in $S_1$ are placed to the left of $\ell$ whereas the vertices in $S_2$ can be placed to both sides of $\ell$.
A simple procedure is used to obtain the partition.
In the partition $V$, the vertices are scanned from left to right. These vertices along with all their neighbors in $U$ are included in $S_1$. The process is stopped when $S_1$ has at least $\frac{n}{2}$ vertices.
Consider the situation when before scanning a vertex $v_i$, there are less than $\frac{n}{2}$ vertices in $S_1$. After $v_i$ is scanned, there are more than $\frac{n}{2}$ vertices in $S_1$. Note that all the new vertices added
to $S_1$ while scanning $v_i$ are the pendant vertices within $S_1$, i.e. these vertices have only one edge incident to them in the subgraph induced on the vertices in $S_1$. 
All other edges incident to these vertices cross $\ell$. These vertices are called the {\it terminal vertices}. The partition obtained by this procedure has the following properties.
\begin{enumerate}
 \item If any edge incident to a vertex in $S_1$ has its other end point to the left
of $\ell$, then the corresponding vertex must be in $S_1$.
 \item For any vertex in $S_1$, there is at least one edge incident to another vertex in $S_1$, i.e. both the vertices defining the edge are placed to the left of $\ell$.
\end{enumerate}

Let us now consider the edges with one end point in $S_1$ and the other end point in $S_2$. All such edges must cross the line $\ell$ by property (1) of the partition.
Lemma~\ref{pl} can be applied to count such edges due to property (2) of the partition.
By Lemma~\ref{pl}, the maximum number of these edges is at most the summation of the number of vertices in $S_1$ and the number of vertices (in $S_2$) that are placed to the right of $\ell$.
Thus, the number of such edges is at most $n-1$. Let $\mathcal{T}(P)$ denote the maximum number of edges a $PRBG$ on a vertex set $P$ can have, then $\mathcal{T}(S_1 \cup S_2) \le \mathcal{T}(S_1) + \mathcal{T}(S_2) + n-1$.
Terminal vertices can be dropped from $S_1$ as they have only one edge incident to them. Thus, both the partitions $S_1$ and $S_2$ have at most $\frac{n}{2}$ vertices.
Now the same procedure can be independently applied to count the edges in $S_1$ and $S_2$ recursively. Thus, $\mathcal{T}(U \cup V) = n \log n + O(n)$. It proves that the number of edges in $G$ is at most $n \log n + O(n)$.
\begin{figure}[h]
\begin{center} \label{u5}
 \includegraphics [scale=0.4]{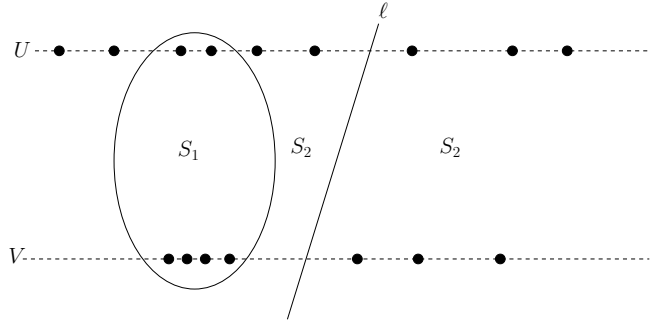}
 \caption{Partition of the point set}
\end{center}
\end{figure}

A matching lower bound can be obtained by a similar 0-1 matrix proposed in ~\cite{tds}.
\qed
\begin{theorem}
 Any unit distance graph on a convex independent point set with $n$ points has at most $2n \log n + O(n)$ edges.
\end{theorem}
\proof
Recall from Section~\ref{sc2} that a $UDG$ on a convex point set can be decomposed into two graphs in $G_{UDG}$ by removing at most $3n$ edges.
A graph in $G_{UDG}$ is also a $PRBG$ that has at most $n \log n + O(n)$ edges. Thus, it concludes that any unit distance graph on a convex point set
has at most $2n \log n + O(n)$ edges.
\qed

Similarly, the following theorem can be established for the $LGGs$.
\begin{theorem}
 Any locally Gabriel graph on a convex independent point set with $n$ points has at most $2n \log n + O(n)$ edges.
\end{theorem}
\section{Number of edges in unit distance graphs}\label{sc5}
In this section, first we show that Class $G_{UDG}$ is a strict sub class of the class $G_{LGG}$. It also establishes that the class of $UDGs$ on convex point sets
is a strict sub class of the $LGGs$ on convex independent point sets. Then, we show that a graph in $G_{UDG}$ and therefore a $UDG$ on a convexly independent
point set has at most a linear number of edges. The arguments for the same have a following high level approach. We define a subgraph of the
class $G_{UDG}$ called {\it module} and argue that any graph in $G_{LGG}$ is formed by interconnected modules. Each module has a linear number of
edges and Subsequently, the number of interconnecting edges is linear too.\\
Recall that two graphs $G'_1$ and $G'_2$ in $G_{UDG}$ are obtained from a $UDG$ on a convexly independent point (refer to Section~\ref{sc2}). In this section,
all the arguments are presented for $G'_1$ type of $G_{UDG}$. Symmetric arguments can be made for $G'_2$ in support of all the required proofs.
\begin{lemma} \label{ul}
 The class $G_{UDG}$ is a strict sub class of the class $G_{LGG}$.
\end{lemma}
\begin{figure}[ht]
\centering
\includegraphics[scale=0.3]{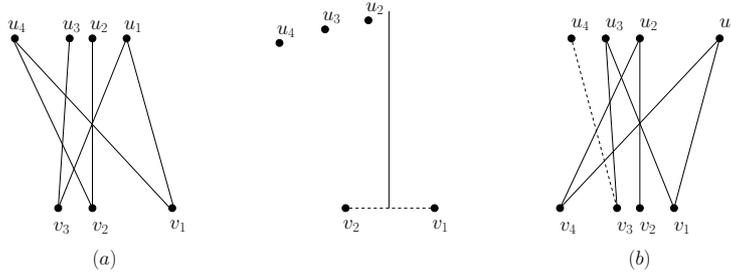}
\caption{A forbidden pattern in $G_{UDG}$}
\label{u6}
\end{figure}
\proof
Recall that every $UDG$ is an $LGG$ by Remark~\ref{uisl} and subsequently a graph in $G_{UDG}$ is also in $G_{LGG}$. The stated 
Lemma is proved by contradiction. We show a simple example of a graph that is forbidden in the class $G_{UDG}$ and can be a member of the class $G_{LGG}$.
Consider the graph shown in Figure~\ref{u6}($a$), we show that this graph cannot be a member of the class $G'_1$ type of $G_{UDG}$
(refer to Section~\ref{sc2} for the definition), i.e. it cannot be a $UDG$ with all points in convex position satisfying the acute angle property in $V$ (by definition of $G'_1$).
In the quadrilateral $v_1u_1u_3v_3$, by the definition of $G'_1$ type of $G_{UDG}$ and convexity, $\angle u_3v_3v_2 < \frac{\pi}{2}$. It can be observed by convexity that $\angle u_2v_1v_3 < \frac{\pi}{2}$.
By the property of isosceles triangles, $\angle v_3u_3u_1$ and $\angle u_1v_1v_3$ are acute. Therefore, in the quadrilateral $v_1u_1u_3v_3$, $\angle v_1u_1u_3$ is greater than $\frac{\pi}{2}$.
By convexity, $\angle v_1u_1u_2$ is greater than $\frac{\pi}{2}$. Therefore, $\overline{u_2v_1} > \overline{u_1v_1}$. Since $ \overline{u_2v_2}$ has unit length, $\overline{u_2v_1}$ has length more than unity.
The locus of the points equidistant from $v_1$ and $v_2$ is the perpendicular bisector to the line joining these points as shown in Figure~\ref{u6}($a$).
Observe that $\overline{u_4v_1} > \overline{u_4v_2}$. Thus, $v_1$ and $v_2$ both cannot have an edge incident to $u_4$.
Observe that this graph can be a member of the class $G_{LGG}$ when $U$ and $V$ are monotonic sequences in opposite quadrants.\\
Now we show that a symmetrically opposite forbidden pattern (across partitions $U$ and $V$ as shown in Figure~\ref{u6}($b$)) is also forbidden.
Recall that in order to obtain a graph in $G_{UDG}$, for every vertex $v_i \in V$ the edge incident to the vertex in $U$ with the highest
order is removed (refer to {\it left trimming} in Section~\ref{sc2}). Let $u_4$ be the vertex with the highest order with an edge incident to $v_3$. Thus, in the quadrilateral $u_1v_1v_3u_4, \angle u_1v_1v_3 > \frac{\pi}{2}$.
Therefore, both the vertices $u_1$ and $u_2$ cannot have edges incident to $v_4$ as shown in Figure~\ref{u6}($b$).
\qed
\begin{corollary}\label{forbidcor}
 The patterns shown in Figure~\ref{ffig} are forbidden in $G_{UDG}$. Note that the dotted edges indicate any possible forward path. 
\end{corollary}

\begin{figure}[ht]
\begin{center}
 \includegraphics [scale=0.5]{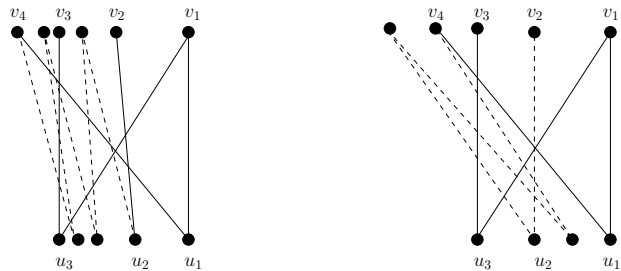}
 \caption{\small Forbidden patterns in $G_{UDG}$}
 \label{ffig}
\end{center}
\end{figure}
\proof
The proof follows the same argument as in Lemma~\ref{ul}. The distance between $u_1$ and $v_2$ is larger than the distance between $u_2$ and $v_2$.
Thus, $\overline{u_1v_4} > \overline{u_2v_4}$. It implies that the configurations shown in the Figure~\ref{ffig} are not feasible.
\qed
\begin{remark}
 For a graph $G = (U,V,E)$ in $G_{UDG}$, if two vertices $u \in U$ and $v \in V$ are spanned by some tree $T_l(u_0)$ such that $(u,v) \notin T_l(u_0)$,
 then $(u,v) \notin E$. 
\end{remark}
Now we present an improved bound for the maximum number of edges in $UDGs$ on convex independent point sets.
If there exists a vertex $v_0 \in V$ for a graph $G(U,V,E)$ in $G_{UDG}$ such that apart from all the vertices in $T_{l}(v_0)$, $\forall v \in V, v < v_0$
and $\forall u \in U, u < u_0$ where $u_0 \in U$ is a vertex in $T_l(v_0)$ with the least order. Assume that all the edges in this graph apart from the edges in 
$T_l(v_0)$ are crossing the edge $(u_0,v_0)$. Thus, by Lemma~\ref{pl} (crossing lemma) the number of these edges is bounded by $|U|+|V|$. Also assume
that all the vertices with order less than $v_0$ or $u_0$ have at least one edge incident to some vertex in $T_l(v_0)$. The number of edges
between these vertices (order less than $v_0$ or $u_0$) is linear by Corollary~\ref{plc}.
This type of $G_{UDG}$ is called a {\it modular $G_{UDG}$} or a {\it module}.
\begin{corollary}\label{linmod}
 The number of edges in a module $G=(U,V,E)$ is  bounded by $O(|U|+|V|)$.
\end{corollary}
The subgraph of a module induced by the vertices and edges in $T_l(v_0)$ is called the {\it core} of a module and the remaining edges
are called the {\it auxiliary edges}. The vertices to which the auxiliary edges are incident (not in the core) are called {\it auxiliary vertices}.
A high level of our approach is to show that any graph in the class $G_{UDG}$ can be decomposed into interconnected modules. 
Let us describe two kinds of orientations for a given pair of disjoint modules. In the first orientation the modules are linearly separable. In such
\begin{figure}[ht]
\begin{center}
 \includegraphics [scale=0.5]{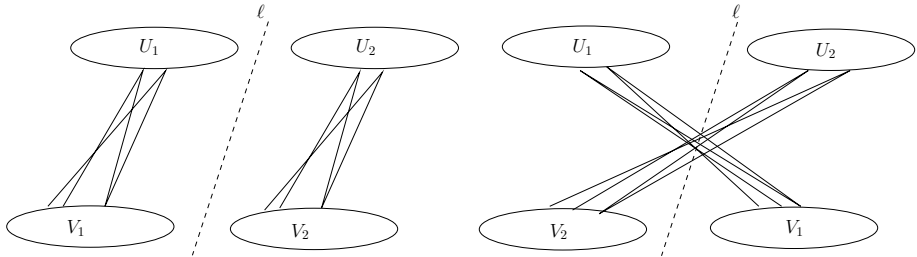}
 \caption{\small Linearly separable modules and cross separable modules}
 \label{u11}
\end{center}
\end{figure}
a case, there exists a separator line such that all the vertices of both the modules lie on the opposite sides of this line, i.e. two modules
$G_1 = (U_1,V_1,E_1)$ and $G_2 = (U_2,V_2,E_2)$ are linearly separable if $\forall u_i \in U_1$ (resp. $\forall v_i \in V_1$) and
$\forall u_j \in U_2$ (resp. $\forall v_j \in V_2$) either $u_i > u_j$ and $v_i > v_j$ or $u_i < u_j$ and $v_i < v_j$.
Further, $G_1$ and $G_2$ are said to be partially linearly separable if $\forall u_i \in U'_1$ (resp. $\forall v_i \in V'_1$) and
$\forall u_j \in U_2$ (resp. $\forall v_j \in V_2$) either $u_i > u_j$ and $v_i > v_j$ or $u_i < u_j$ and $v_i < v_j$ where
$U'_1$ and $V'_1$ are the sets corresponding to the core vertices in $G_1$ .
On the contrary, two modules $G_1 = (U_1,V_1,E_1)$ and $G_2 = (U_2,V_2,E_2)$ are cross separable if $\forall u_i \in U_1$ (resp. $\forall v_i \in V_1$) and
$\forall u_j \in U_2$ (resp. $\forall v_j \in V_2$) either $u_i > u_j$ and $v_i < v_j$ or $u_i < u_j$ and $v_i > v_j$.

Now we introduce a procedure called {\it partitioning}. If a module is partitioned along a line $\ell$, then the module is separated into 
smaller units such that for any of the resultant module either all the vertices lie on one side of $\ell$ or the vertices in each partition (opposite
partitions of bipartite graphs) lie on the opposite sides of $\ell$.

\begin{lemma}\label{graphpart}
 A module can be partitioned along any partition line.
\end{lemma}
\begin{figure}[ht]
\begin{center}
 \includegraphics [scale=0.4]{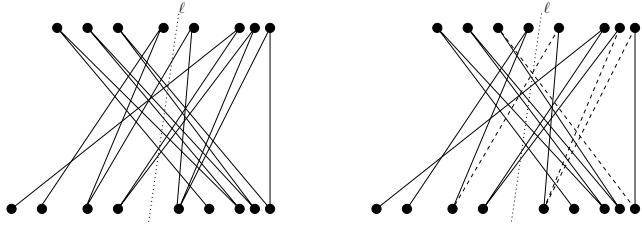}
 \caption{\small Partitioning a module}
 \label{udgfpart}
\end{center}
\end{figure}
\proof
Observe the forwards paths in a left tree ($T_l(v)$ for some vertex $v$). Note that all the forwards path in a left tree are cross separable. Let $\ell$ be the partition line.
The forward paths placed to the opposite side of $\ell$ can be abstracted as two separate linearly separable modules. Consider the edges crossing $\ell$.
The edges in a forward path intersecting $\ell$ form another module such that the vertices in both the bipartioins are placed on the opposite side of $\ell$.
Note that some of the edges in the original module now have their both constituting vertices in the different modules.
For an example, refer to the Figure~\ref{udgfpart}$(a)$ for the core of a module and a partition line $\ell$. The resultant modules after partitioning
are shown in Figure~\ref{udgfpart}$(b)$. The dotted edges are the edges between the vertices of different modules.
\qed

Consider the case when in a graph in $G_{UDG}$, there is a pair of overlapping modules, i.e. they are neither (partially) linearly separable nor cross separable. Such modules
can be partitioned down further to ensure that any pair of modules is either (partially) linearly separated or cross separated.

Consider two cross separated modules as shown in Figure~\ref{u11}. Note that there exist any edges between $U_1$ and $V_2$, then these
modules are partitioned further to ensure that no such edges exist or a big module can be abstracted from these vertices.
Edges can exist between $U_2$ and $V_1$ though. Each vertex in $U_2$ or $V_1$ can have at most one such
edge incident to it (refer to Lemma~\ref{fusion}). The union of two cross separable modules with such connecting edges is called a fused module and 
the abstracting a fused module from two basic modules is called {\it fusing}. The edges across two modules comprising of a fused module
are called {\it fusing edges}. These terms are also used to combine two (partially) linearly separable modules in a hybrid fused module.
Let us consider two cross separable modules $G_1=(U_1,V_1,E_1)$ and $G_2=(U_2,V_2,E_2)$ such that $\forall u_i \in U_1 > \forall u_j \in U_2$
and $\forall v_i \in V_1 < \forall v_j \in V_2$. Let us consider the possible 
adjacencies between $V_1$ and $U_2$. We argue that the set of such edges form a matching, i.e. no vertex has more than one edges incident to it.
\begin{lemma}\label{corefusion}
 In a module $G = (U,V,E)$, two vertices $u_1$ and $u_2$ (resp. $v_1$ and $v_2$) cannot have an edge incident to $v_0 < \forall v_i \in V$ 
 (resp. $u_0 < \forall u_i \in U$) unless both $u_1$ and $u_2$ (resp. $v_1$ and $v_2$) are the core vertices in $G$.
\end{lemma}
The statement of Lemma~\ref{corefusion} can also be interpreted as Lemma~\ref{fusion}.
\begin{lemma}\label{fusion}
 For cross separable modules $G_1=(U_1,V_1,E_1)$ and $G_2=(U_2,V_2,E_2)$ such that $\forall u_i \in U_1 > \forall u_j \in U_2$ and 
 $\forall v_i \in V_1 < \forall v_j \in V_2$, there can be only one-to-one adjacencies between $V_1$ and $U_2$ unless two or more core vertices
 have edges incident to a common vertex.
\end{lemma}
\begin{figure}[ht]
\begin{center} 
 \includegraphics [scale=0.5]{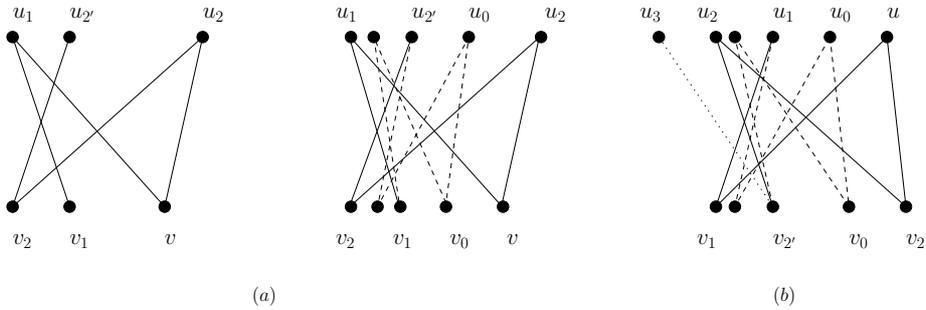}
 \caption{\small Edges between two cross separable modules}
 \label{u12}
\end{center}
\end{figure}
\proof
First we prove that in a module $G_1=(U_1,V_1,E_1)$ two vertices $u_1 \in U_1$ (a core vertex) and $u_2 \in U_1$ (an auxiliary vertex) cannot have an edge incident to $v$ such that $v < v_i 
\quad \forall v_i \in V_1$. Let us prove it by contradiction. Let us assume that $u_1$ and $u_2$ have an edge incident to $v$ as 
shown in Figures~\ref{u12}($a$).
Let $v_2$ be the core vertex with an auxiliary edge incident to $u_2$. Since $v_2$ is a core vertex, it also has an edge incident to at least one
core vertex $u_{2'} > u_2$. By path restricted property, $u_1 > u_{2'} > u_2$. Since $u_1$ is a core vertex, it has an edge incident to a core 
vertex $v_1$. Again by path restricted property, $v_2 > v_1 > v$. Since $v_1$ and $v_2$ are the core vertices in the same module, there exists
another path between them. Let us assume w.l.o.g. that this path passes through $u_{2'}$.
Note that two core vertices are always connected by a left tree. Thus, this tree provides a path between
$v_1$ and $v_2$ (subsequently $u_{2'}$). Let $v_{1'}$ be the immediate neighbor of $u_{2'}$. If $v_1 = v_{1'}$ or $v_1 > v_{1'} > v$,
then both $u_2$ and $u_2'$ cannot have an edge incident to $v_2$ (refer to Lemma~\ref{ul}). Similarly, if $v_1$ has an edge incident to
$u_{2'}$ or a vertex between $u_{2'}$ and $u_2$, then $v$ and $v_1$ both cannot have an edge incident to $u_1$ (refer to Lemma~\ref{ul}).
Thus, $u_{2'} < u_{1'} < u_{1}$ and $v_1 < v_{1'} < v_2$. Therefore, there exist vertices $u_0$ and $v_0$ such that there exist forward paths
with ranges $\{\langle v_0,v_1 \rangle,\langle u_0,u_{1'}\rangle \}$ and $\{\langle v_0,v_{1'}\rangle,\langle u_0,u_{2'}\rangle\}$ respectively. Now we show that this configuration is not feasible.
In the quadrilateral $v_2u_{2'}u_2v$, $\angle vu_{2'}u_2 > \frac{\pi}{2}$. Also note that the distance between $v_0$ and $u_1$
is greater than unity. Thus, $v$ cannot have an edge incident to $u_1$ since $\overline{v_0u_1} < \overline{vu_1}$.\\
Now we argue the same claim for symmetrically opposite case where two vertices $v_1 \in V_1$ (a core vertex) and $v_2 \in V_1$ 
(an auxiliary vertex) cannot have an edge incident to $u$ such that $u < u_i \quad \forall u_i \in U_1$ as shown in Figures~\ref{u12}($b$).
Let $u_3$ be the vertex with the highest order incident to $v_{2'}$ before {\it left trimming} (refer to Section~\ref{sc2}). 
Thus, in the quadrilateral $uv_2v_{2'}u_3, \angle uv_2v_{2'} > \frac{\pi}{2}$. Thus, similarly $u$ cannot have an edge incident to $v_1$.\\
By Corollary~\ref{plc} no two auxiliary vertices can have an edge incident to the same vertex outside a module.
Thus, amongst the edges incident between $G_1$ and $G_2$, any vertex in either module has at most one edge incident to it.
\qed

Observe the following Corollary from Lemma~\ref{corefusion}.
\begin{corollary}\label{distcol}
 If all the edges but $(u_2,v)$ (resp. $(u,v_2)$) as shown in Figure~\ref{u12} exist in a $G'_1$ type of $G_{UDG}$, then the distance
 between $u_2$ and $v$ (resp. $u$ and $v_2$) is less than unity, i.e. less than the distance between two vertices inducing an edge.
\end{corollary}

\begin{lemma}\label{linpart}
Let $G_1=(U_1,V_1,E_1)$ and $G_2=(U_2,V_2,E_2)$ be two (partially) linearly separable modules such that the core vertices of $G_2$ have higher order
than the same in $G_1$. Only one vertex in $G_1$ can have edges incident to the core vertices of $G_2$ or the same vertex has at most one edge incident to an auxiliary
vertex in $G_2$. 
\end{lemma}
\proof
Let us consider the edges between two linearly separable modules. Observe that no auxiliary vertex in $G_1$ has an edge incident to a vertex in $G_2$.
Only one core vertices in $G_1$ with the highest order (from either partition) can have an edge incident to a vertex in $G_2$. Note that if
such an edge exists, then $G_1$ has a star shaped core, i.e. one of the bipartition has only one vertex (say $u_0 \in U_1$) with edges incident to all vertices
of the module in $V_1$. Otherwise, $G_1$ and $G_2$ will not be linearly separable. Let us consider all the edges incident to
$u_0$ with the vertices in $G_2$. If there is only one such edge, then the statement of the lemma is proved. Otherwise, if $u_0$ has an edge
incident to at least one auxiliary vertex, then $u_0$ can have only edge incident to the vertices in $G_2$ by Lemma~\ref{corefusion}.
Let us consider the situation when a vertex from the other bipartition (say $v_0 \in V_1$) has an edge incident to $G_2$. Note that $(u_0,v_0) \in E_1$. Both $u_0$ and $v_0$ cannot have
edges incident to the vertices of a forward path in $G_2$ by Corollary~\ref{lggpbg}. Let $u_0$ has an edge incident to $v_1$ and $v_0$ has an edge incident
to $u_1$. Thus, $u_1$ and $v_1$ are the vertices in the different branches of the core of $G_2$ as shown in Figure~\ref{linpartfig}. This configuration
is not feasible by Lemma~\ref{ul}.
\begin{figure}[ht]
\begin{center} 
 \includegraphics [scale=0.5]{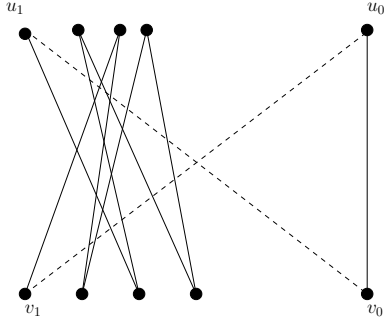}
 \caption{\small Edges between two cross separable modules}
 \label{linpartfig}
\end{center}
\end{figure}

Consider the case when $G_1$ and $G_2$ are partially linearly separable. Note that no vertex in $G_1$ can have an edge incident to a vertex
in $G_2$ with less order than any vertex of $G_1$ by path restricted property. Thus, it proves that linearly separable modules can have
at most one edge incident between them.
\qed

We can extend Lemma~\ref{linpart} to the following lemma.
\begin{lemma}
Two linearly separated fused modules $G_1 = (U_1,V_1,E_1)$ and $G_2 = (U_2,V_2,E_2)$ such that $\forall u_i \in U_1 < \forall u_j \in U_2$
(resp. $\forall v_i \in V_1 < \forall v_j \in V_2$). Then only those vertices in $G_1$ with no fusing edge incident to can have at most one edge
incident to any vertex in $G_2$.
\end{lemma}
\begin{lemma}\label{fuseset}
Let a set of cross separable modules $\{G_1 =(U_1,V_1,E_1),\\ G_2=(U_2,V_2,E_2), \ldots, G_k=(U_k,V_k,E_k)\}$ are fused successively
such that for any $u_i \in U_i , u_j \in~U_j,\\ v_i \in V_i$ and $v_j \in V_j,
u_i > u_j$ and $v_i < v_j$ for $i < j$. The number of fusing edges in such an arrangement is linear.
\end{lemma}
\begin{figure}[ht]
\begin{center} 
 \includegraphics [scale=0.45]{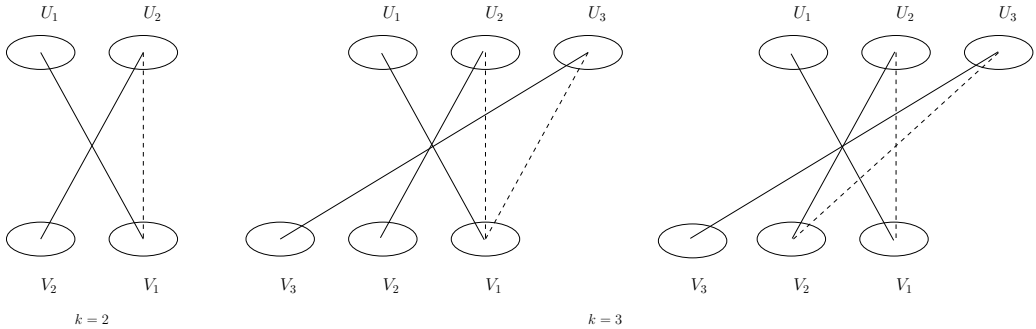}
 \caption{\small Fusing edges across multiple three separable modules}
 \label{fusedmods}
\end{center}
\end{figure}
\proof
Let us give an incremental proof for the stated lemma. Let us begin with $k = 2$. Note that if two or more core vertices of a module
have an edge incident to a vertex in another module, these edges are not counted as the fusing edges and a module has a linear
number of edges with such edges accounted by Corollary~\ref{plt}. There can be only one-to-one adjacencies between $V_1$ and $U_2$
by Lemma~\ref{fusion}.\\ 
If $k = 3$ and another module $G_3$ is fused, then the vertices of $G_2$ and $G_3$ can also have only one-to-one
adjacencies by Lemma~\ref{fusion}. We show that if $v_1 \in V_1$ has an edge incident to $u_3 \in U_3$, then either $v_1$ has no fusing edge incident
to any vertex in $U_2$ (in $G_2$) or $u_3$ has no fusing edge incident to any vertex in $V_2$ (in $G_2$). It can be easily argued by
Lemma~\ref{corefusion}. Edge $(u_2,v_1)$ can also be realized as an {\it auxiliary edge} of module $G_2$. Thus, if the edge $(u_2,v_1)$ exists, then $v_1$ and any vertex
$v_2 \in V_2$ cannot have edges incident to $u_3 \in U_3$. Thus, the stated lemma holds true for $k = 3$. Note that $\overline{v_2u_3} > \overline{v_1u_3}$ 
by Corollary~\ref{distcol}. Note that the fusing edges can be counted by charging at most one edge to each vertex.\\
Let us now consider $k = 4$. Since the lemma holds true for $k = 3$, the fused module formed by $G_1 , G_2$ and $G_3$ denoted as $G_{1,2,3}$ has a
linear number of fusing edges with the pattern discussed above. Similarly, $G_{2,3,4}$ has a linear number of edges. Now consider the adjacencies between $V_1$
(in $G_1$) and $U_4$ (in $G_4$). Refer to Figure~\ref{fuse4} for pictorial representation. If $v_1 \in V_1$ and $u_4 \in U_4$ have no fusing edges incident to $G_2$ or $G_3$,
then the edge $(u_4,v_1)$ is feasible.
Let us consider the situation when $v_1 \in V_1$ and $u_4 \in U_4$ have adjacencies in $G_2$ and $G_3$. If $v_1 \in V_1$ has an
edge incident to a vertex in $U_2$ and some $v_2 \in V_2$ has an edge incident to $U_3$ and some $v_3 \in V_3$ has an edge incident to $u_4 \in U_4$.
Note that $\overline{u_4v_3} > \overline{u_4v_2}$ by Corollary~\ref{distcol}. Also $\overline{u_4v_2} > \overline{u_4v_1}$. Thus, in this 
condition $v_1 \in V_1$ and some $v_3 \in V_3$ cannot have edges incident to $u_4 \in U_4$. 
Note that there can be other cases where $v_1 \in V_1$ has a path to $v_4 \in V_4$ by fusing edges through either $V_2$ or $V_3$.
In that case, again $v_1$ cannot have an edge incident to $v_4$. However, a vertex in $V_2$ or $V_3$ (not having adjacencies with $V_1$) 
can have an edge incident to $V_4$.\\
\begin{figure}[ht]
\begin{center} 
 \includegraphics [scale=0.6]{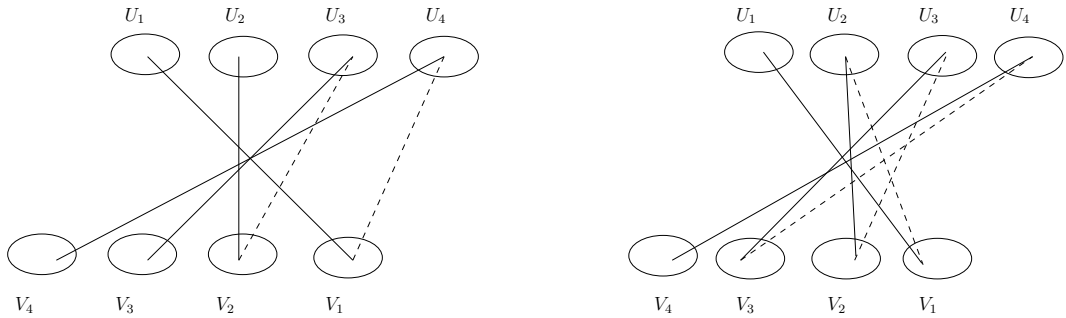}
 \caption{\small Fusing edges across four cross separable modules}
 \label{fuse4}
\end{center}
\end{figure}
Note that when every module has edge(s) incident to the successive module, then the
same argument can be extended further for the higher values of $k$. Thus, the stated lemma holds true in this case.
Let us consider the case when this condition does not hold. Recall when $k=3$, if $G_2$ and $G_3$ have no {\it fusing edges} incident across,
then both these modules can have {\it fusing edges} incident to $G_1$. Similarly if $G_1$ and $G_2$ have no {\it fusing edges} incident across,
then both these modules can have {\it fusing edges} incident to $G_3$. Let us analyse the general case where $U_i$ (in $G_i$) and $U_j$ (in $G_j$) ($j > i$)
do not have a path by {\it fusing edges} through $G_{i+1},\ldots,G_j$. Thus, $U_i$ and $U_j$ can have edges incident to a common module $G_k$ ($k > j$).
Note that by Lemma~\ref{pl} such edges can only form a tree with a linear number of edges. Note that $U_i$ or $U_j$ can also have edges
incident to more than one modules. Observe that in this tree except for the edges to leaves in the other partition, the number of edges 
is linear (in terms of $|U_i| + |U_j|$). Let us now focus on the edges to the leaves in other partiton. If a vertex in $U_i \cup U_j$ has
edges incident to more than one vertices, then note that these vertices are in different modules by Lemma~\ref{fusion}. By a similar argument on
the other partition, these edges are part of another tree. Note that this argument is applicable for an arrangement of more than two modules as well.
Thus, if a module has no edges incident to the successive module, still the total number of fusing edges remains linear. It proves the stated lemma.
\qed

A $G \in G_{UDG}$ can be partitioned either into a set of modules such that any two modules are (partially) linearly separated or cross separated.
Such an alignment of the modules is called a {\it module sequence}.
Thus, a bottom up hierarchy of the modules can be created where the modules are grouped and fused successively leading to a graph in $G_{UDG}$.
Now we argue that the number of fusing edges in a module sequence is linear since every
vertex has at most one fusing edge incident to it (refer to Lemma~\ref{linpart} and Lemma~\ref{fuseset}).
Recall that the number of {\it auxiliary edges}
in a module is linear by Corollary~\ref{plt}. A module has a linear number of edges by corollary~\ref{linmod}. Thus, a graph in the class $G_{UDG}$ has a linear number of edges.
A $UDG$ on convex point sets $G$ can be partitioned into two graphs $G_1 \in G_{UDG}$ and $G_2 \in G_{UDG}$. Thus, we conclude that a unit distance
graph on a convex point set has a linear number of edges.
\begin{theorem}\label{udglin}
A $UDG$ on a convex point set with $n$ vertices has $O(n)$ edges.
\end{theorem}
\section{Concluding Remarks}In this note, we defined a family of bipartite graphs known as the path restricted ordered bipartite graphs. We also showed that these graphs can be obtained from various geometric graphs
on convex point sets. We studied various structural properties of these graphs and showed that a path restricted ordered bipartite graph on $n$ vertices has $O(n \log n)$ edges and this
bound it tight. The same upper bound was already known for the unit distance graphs and the locally Gabriel graphs on convex point sets. However, the best known lower bound known
to the edge complexity on these graphs for convex point sets is $\Omega(n)$. We improved the upper bound for unit distance graphs to $O(n)$.
The problem of bridging the gap in the bounds remains an open for the locally Gabriel graphs on a convexly independent point set.\\

\bibliographystyle{amsplain}
\bibliography{references}
\end{document}